\documentclass[a4paper, reqno]{amsart}
\usepackage{a4,wasysym}
\usepackage{amsmath,graphicx, amssymb,color,ulem,bbm}

\newtheorem{theorem}{Theorem}

\newtheorem{lemma}[theorem]{Lemma}

\begin{document}

\title[Modelling rapid evolution]
{Modelling effects of rapid evolution on persistence and stability in structured predator-prey systems}

\author[J\'{o}zsef Z. Farkas]{J\'{o}zsef Z. Farkas}
\address{J\'{o}zsef Z. Farkas, Division of Computing Science and Mathematics, University of Stirling, Stirling, FK9 4LA, United Kingdom }
\email{jozsef.farkas@stir.ac.uk}

\author[A. Yu. Morozov]{A. Yu. Morozov}
\address{A. Yu. Morozov, Department of Mathematics, University of Leicester, Leicester, LE1 7RH, UK; and  
 Shirshov Institute of Oceanology, Moscow, 117997, Russia}
\email{}

\subjclass{47A75, 45K05, 92D40}
\keywords{Integro-differential equations,  structured populations,
population persistence, stability analysis, spectral theory of operators.}
\date{\today}

\begin{abstract}
In this paper we explore the eco-evolutionary dynamics of a predator-prey model, where the prey population is structured according to a certain life history trait. The trait distribution within the prey population is the result of interplay between genetic inheritance and mutation, as well as selectivity in the consumption of prey by the predator. The evolutionary processes are considered to take place on the same time scale as ecological dynamics, i.e. we consider the evolution to be rapid. Previously published results show that population structuring and rapid evolution in such predator-prey system can stabilize an otherwise globally unstable dynamics even with an unlimited carrying capacity of prey. However, those findings were only based on direct numerical simulation of equations and obtained for particular parameterizations of model functions, which obviously calls into question the correctness and generality of the previous results. The main objective of the current study is to treat the model analytically and consider various parameterizations of predator selectivity and inheritance kernel. We investigate the existence of a coexistence stationary state in the model and  carry out stability analysis of this state. We derive expressions for the Hopf bifurcation curve which can be used for constructing  bifurcation diagrams in the parameter space without the need for a direct numerical simulation of the underlying integro-differential equations. We  analytically show the possibility of stabilization of a globally unstable predator-prey system with prey structuring.  We prove that the coexistence stationary state is stable when the saturation in the predation term is low. Finally, for a class of kernels describing genetic inheritance and mutation we show that stability of the predator-prey interaction will require a selectivity of predation according to the life trait.
\end{abstract}
\maketitle

\section{Introduction}

 It is becoming increasingly apparent that biological evolution of organisms' life history traits can occur on a similar time scale to their population dynamics \cite{Thompson1998,Yoshida_etal2003, Duffy_2007, Matthews_etal2011, Fussmann_Gonzalez2013}, and that such rapid evolution can shape the ecological dynamics of interacting species \cite{Ellner_2011, Jone_etal2009, Morozov2013}. In particular, predator-prey and host-parasite interactions can be dramatically affected by rapid genetic variation within the populations \cite{Dube_2002, Hairston_2008, Reznick_2008, Johnson_2009, Jone_etal2009}, and we need to take this fact into account in food-web and epidemiological models. For instance, it was demonstrated in \cite{Jones_2007, Jone_etal2009} that rapid genetic variation can enlarge the period and modify the  phase relations in predator-prey cycles.   Recently, using a parsimonious eco-evolutionary predator-prey model, Morozov \textit{et al.} have shown that the interplay between a fast genetic variation within the prey and the selectivity of consumption of the prey by its predator can suppress large amplitude oscillations in the system \cite{Morozov2011, Morozov2013}. Interestingly, the reported stabilization could occur even in the case when the carrying capacity of the prey is infinitely large. In contrast, the same system with non-changing life history traits of species would be globally unstable, see e.g. \cite{Oaten_Murdoch_1975}. Thus, a rapid evolution could potentially provide an extra mechanism of persistence for trophic interaction in a highly eutrophic environment (known also as the top-down control) \cite{Morozov2011, Morozov2013}. This mechanism would answer the long standing open question in theoretical ecology about possible solutions of the paradox of enrichment: what factors can stabilize predator-prey interactions in the case the supply of resource for prey is large?  \cite{Rosenzweig1971, Oaten_Murdoch_1975, Abrams_1996, Morozov_etal_2011}.

  The previous results on stabilization in the eco-evolutionary predator-prey model with prey structuring in a highly eutrophic ecosystems were obtained by directly simulating the underlying model equations for particular parameterisations of the model functions \cite{Morozov2013}. Obviously, this cannot be considered as a rigorous proof of stability. Furthermore, considering other functional parameterizations of the inheritance kernel as well as the selectivity of predation can potentially affect the results obtained only using simulations. Thus, the central question is whether or not the main conclusions in \cite{Morozov2011, Morozov2013} on the importance of within population structuring on persistence and stability will be generic.

  The main objective of this paper is to explore in detail the mathematical properties of the eco-evolutionary predator-prey model \cite{Morozov2011, Morozov2013}. Our analytical investigation is centred around two main issues: the existence of non-trivial stationary state(s) of the model (assuring the coexistence of the prey and predator), and the stability of the stationary state(s) for a general family of functions (kernels) describing genetic inheritance within the prey.  We also establish global existence and positivity of solutions of the model.

Mathematically we consider the following system of integro-differential equations (which is a generalization of that in \cite{Morozov2013}):
\begin{eqnarray}
\label{eq:Zoo}
\frac{\partial}{\partial t} z(x,t) &=& \int_{x_1}^{x_2}k(x,y)z(y,t)\, dy - F(t)\frac{a(x)z(x,t)}{1+\beta \int_{x_1}^{x_2}z(x,t)\,dx},\\
\label{eq:Fish}
\frac{d}{d t}F(t) &=&F(t)\left(\frac{K\int_{x_1}^{x_2}a(x)z(x,t)\, dx}{1+\beta\int_{x_1}^{x_2}z(x,t)\, dx}-\mu\right),
\end{eqnarray}

\noindent where $z(x,t)$ describes the distribution of prey individuals across the cohorts according to a certain life history trait $x$, which can vary between $x_1$ and $x_2$. Hence the integration of $z(x,t)$ over $[x_1,x_2]$ gives the total biomass of prey in the population, while $F(t)$ is the total biomass of the predator at time $t$. Here we consider the structuring life history trait $x$ to be an abstract parameter: in practice this can be the growth rate, body length (at the adult stage), individual mobility \cite{Petrovskii_2010} or it can be the animals personality \cite{Wolf_2012}.  The life trait is genetically inherited, i.e. it does not change throughout the lifetime. For the sake of simplicity we assume that the predator population is genetically uniform (i.e. there is no structuring of the predator population).

The integral term in equation (\ref{eq:Zoo}) stands for the growth rate of prey due to the reproduction of all cohorts. The kernel $k(x,y)$ describes genetic inheritance and mutation. Each cohort can produce offspring with the life trait within  $[x_1,x_2]$. The contribution to the growth rate of cohort $x$ due to reproduction of the group of cohorts with life traits in the interval $[y-dy/2,y+dy/2]$ will be given by the product of the population size of this group $z(y,t)dy$ and the kernel $k(x,y)$,with $dy$ being sufficiently small. Thus the kernel function includes both the reproduction rate of group of cohorts centred at $y$ and the probability of producing offspring with trait $x$.  The total contribution of all cohorts to the production of the cohort $x$ is given by integration of $k(x,y)z(x,t)dy$ over the the entire interval $[x_1,x_2]$.

The mortality of prey in the model is due to predation only (i.e. we neglect natural mortality), which is achieved via  a standard Holling type II parametrization, see e.g. \cite{Holling1959, Kot2001, Gentleman2003}, where $\beta$ characterises the saturation of predation at high densities of the prey. The vulnerability of the prey to predation is cohort-dependent and is described by the function $a=a(x)$, i.e. there is prey selectivity by the predator according to the life trait $x$. $K$ is the food conversion coefficient describing the efficiency of transformation of the consumed prey biomass. The predator is subject to a natural mortality $\mu$ which is assumed to be constant.

Model (\ref{eq:Zoo})-(\ref{eq:Fish}) can be considered as a standard Rosenzweig-MacArthur model (see \cite{RosenzweigMacArthur1963, Rosenzweig1971, Allen_2007}) in which genetic structuring and rapid evolution has been incorporated. Our model is similar to that of \cite{Jone_etal2009}, where evolution of cohorts of predator and prey in a chemostat was considered. However, unlike the model in \cite{Jone_etal2009}, here we allow mutations of the clones. In particular, the kernel $k(x,y)$ can be constant which can model the scenario of perfect genetic mixing within the population.

It is of importance to mention that we intentionally consider the carrying capacity of prey population to be infinitely large, i.e. we model predator-prey interactions in a highly eutrophic environment. It is well known that the same model without population structuring is globally unstable and the species persistence is impossible \cite{Oaten_Murdoch_1975}, thus the main question explored using model (\ref{eq:Zoo})-(\ref{eq:Fish}) is whether or not population structuring and rapid evolution (as well as  selectivity of the predator) can eventually stabilize this otherwise unstable system. Note that by rapid evolution we understand here variation in the cohort distribution $z(x,t)$ which occurs on the same time scale as the ecological dynamics (i.e. variation of the total biomass of prey and predator).

Unlike Morozov \textit{et al.} \cite{Morozov2013}, we explore the model for arbitrary parametrisations of the vulnerability to predation $a(x)$ and the inheritance kernel $k(x,y)$. We only make the following natural assumptions regarding the model parameters and $a$ and $k$
\begin{align*}
& 0<a<\overline{a}<\infty,\quad 0<k<\overline{k}<\infty,\quad 0\le\beta,\quad 0<K,\mu.
\end{align*}

Global coexistence and positivity of solutions of model (\ref{eq:Zoo})-(\ref{eq:Fish}), based on the above assumptions for the parameters and integral kernels can be readily established for example using methods developed in \cite{Henry_1981}.

It is worthwhile to point out that unlike in the case of ``standard" physiologically structured population models, see for example the classical size-structured models in \cite{Cushing_1998, FH_2007} (and the references therein), in our model individuals are structured with respect to a genetically inherited property, which implies that there is no transport term in equation \eqref{eq:Zoo}. This allows us to work in the framework of bounded operators. On the other hand, recruitment is represented by an integral operator, which is of infinite rank, in general. Eigenvalue problems involving integral operators are often difficult to analyse. See for example the recent papers \cite{AF_2013,FGH_2010}, where structured population models with distributed recruitment processes (modelled by an integral operator) were analysed. We also note that even though the operators arising from equations \eqref{eq:Zoo}-\eqref{eq:Fish} are bounded, positivity of solutions cannot be immediately established for example from the variation of constants formula.  This is due to the negative feedback between predator and prey populations.

In the next two sections we show the existence of the nontrivial stationary state (coexistence state) of the model for a particular class of kernels, which is actually a separable kernel of a finite rank. Then we derive the characteristic equations to obtain the stability condition of the nontrivial stationary state. Using the obtained generic condition we consider the stability property in a few particular biologically relevant cases, for instance, in the case where the kernel is constant. Finally, we address the important question of the necessity of predation selectivity for the stabilization of model  (\ref{eq:Zoo})-(\ref{eq:Fish}), i.e. whether or not the stability of the coexistence state is possible in the case $a(x) \equiv A$. We summarize our results in the Discussion section.

%%%%%%%%%%%%%%%%%%%%%%%%%%%%%%%%%%%%%%%%%%%%%%%%%%%%%%%%%%%%%%%%%%%%%%%%%%%%%%%%%%%%%%%%%%%%%%%%%%%%%%%%%%%%%%%%%%%%%%%%%%%%%%%%%%%%%%%%%%%%%%%%%%%%%%%%%%%%%%%%%%%%%%%%%%%%%%%%%%%%%%%%%%%%%%%%%%%%%%%%%%%%%%%

\section{Existence of stationary states}
\noindent Note that model (\ref{eq:Zoo})-(\ref{eq:Fish}) always admits the trivial stationary state $(0,0)$.
Throughout the paper we assume that the kernel $k$ is strictly positive, but for the sake of completeness we note that if the kernel $k$ is not strictly positive then our model may admit stationary  states of the form $(0,z^*(x))$, where
$z^*$ satisfies the integral equation
\begin{equation}
\int_{x_1}^{x_2}k(x,y)z^*(y)\, dy \equiv 0.
\end{equation}
For example if $k$ vanishes for every $y$ on some interval $[y_1,y_2]$ then any function $z^*$ which is supported on the interval $[y_1,y_2]$ will be, together with $F^*=0$, a non-trivial stationary state.

In the next two sections we discuss the existence of strictly positive (coexistence or `non-trivial') stationary  states of model (\ref{eq:Zoo})-(\ref{eq:Fish}). We start with the relatively simple case of a separable kernel $k$. In this case
the integral operator describing the recruitment process is of rank one (at most). Then, we will discuss the case of a finite rank approximation of the kernel $k$, and finally,  for a general kernel $k$, we will reformulate the steady state problem as an eigenvalue problem for a bounded operator.

\subsection{Separable kernel}

Assume that $k(x,y)=k_1(x)\bar{k}_1(y)$, for some functions $k_1,\,\bar{k}_1$, i.e. that the kernel is separable. We look for a strictly positive stationary solution $(z^*(x),F^*)$. The
stationary state equations read:
\begin{eqnarray}
\label{sseq1}
k_1(x)\int_{x_1}^{x_2}\bar{k}_1(y)z^*(y)dy &=& F^*\frac{a(x)z^*(x)}{1+\beta\int_{x_1}^{x_2}z^*(x)dx}, \\
\label{sseq2}
\frac{\mu}{K} &=& \frac{\int_{x_1}^{x_2}a(x)z^*(x)dx}{1+\beta\int_{x_1}^{x_2}z^*(x)dx}.
\end{eqnarray}
We introduce the following notation:
\begin{eqnarray}
\label{z1z2}
Z^*\equiv Z_1=\int_{x_1}^{x_2}z^*(x)dx,\quad Z_2=\int_{x_1}^{x_2}a(x)z^*(x)dx,\quad Z_3=\int_{x_1}^{x_2} \bar{k}_1(x)z^*(x)dx.
\end{eqnarray}
Assuming that $a>0$ we obtain from equations (\ref{sseq1}) the following set of scalar equations for the variables
$\left(Z_1,Z_2,Z_3,F^*\right)$.
\begin{eqnarray}
Z_3\int_{x_1}^{x_2}k_1(x)dx & =& F^*\frac{Z_2}{1+\beta Z_1}, \label{scalar1} \\
Z_3\int_{x_1}^{x_2}\frac{k_1(x)}{a(x)}dx & =& F^*\frac{Z_1}{1+\beta Z_1}, \label{scalar2} \\
\int_{x_1}^{x_2}\frac{k_1(x) \bar{k}_1(x)}{a(x)}dx & =& F^*\frac{1}{1+\beta Z_1}, \label{scalar3} \\
\frac{\mu}{K} & =& \frac{Z_2}{1+\beta Z_1}. \label{scalar4}
\end{eqnarray}
One can easily see that a unique positive solution $(z^*(x),F^*)$ of (\ref{sseq1})-(\ref{sseq2}) exists if and only if the 4-dimensional scalar system (\ref{scalar1})-(\ref{scalar4}) above has a strictly positive solution.
We introduce the notation
\begin{equation*}
w=\frac{\mu\int_{x_1}^{x_2}\frac{k_1(x)}{a(x)}\,d x}{K\int_{x_1}^{x_2}k_1(x)\,d x}.
\end{equation*}
Hence, we obtain from (\ref{scalar2})-(\ref{scalar4}) a unique positive solution $Z^*\equiv Z_1=\frac{w}{1-\beta w}$,
as long as $1-\beta w>0$. Substituting this into equation (\ref{scalar3}) yields a unique positive solution $F^*=\frac{1}{1-\beta w}\int_{x_1}^{x_2}\frac{k_1(x)\bar{k}_1(x)}{a(x)}\, d x$. From (\ref{scalar1})-(\ref{scalar4}) we obtain $Z_3=\frac{1}{1-\beta w}\frac{\mu\int_{x_1}^{x_2}\frac{k_1(x) \bar{k}_1(x)}{a(x)}\,d x}{K\int_{x_2}^{x_2}k_1(x)\,d x}$,
and finally from equation (\ref{scalar4}) we obtain $Z_2=\frac{\mu}{K(1-\beta w)}$. Hence the model admits a unique positive stationary state if
\begin{equation}
1>\beta\frac{\mu\int_{x_1}^{x_2}\frac{k_1(x)}{a(x)}\,d x}{K\int_{x_1}^{x_2}k_1(x)\,d x}.\label{sscond1}
\end{equation}
The condition above is satisfied, for instance, in the case when the saturation $\beta$ is not very large and vulnerability to predation $a(x)$ is large enough.

\subsection{Finite rank approximation of the kernel}
The approach above may be extended to the more general case of the kernel:
\begin{equation}
\label{finite_rank}
k(x,y)=\displaystyle\sum_{i=1}^nk_i(x)\bar{k}_i(y),
\end{equation}
i.e. for a finite rank approximation of the kernel $k$ with non-negative functions $k_i,\bar{k}_i,\,i=1,\cdots ,n$.
With $Z_1$ and $Z_2$ as defined in (\ref{z1z2}), and introducing the new variables:
\begin{equation*}
Z_3=\int_{x_1}^{x_2}\bar{k}_1(x)z^*(x)\,d x,\,\cdots,\, Z_{n+2}=\int_{x_1}^{x_2} \bar{k}_n(x)z^*(x)\,d x,
\end{equation*}
the stationary state equations of model (\ref{eq:Zoo})-(\ref{eq:Fish}) read:
\begin{align}
& F^*\frac{a(x)z^*(x)}{1+\beta Z_1}=\sum_{i=1}^nZ_{i+2}k_i(x),\quad \frac{\mu}{K}=\frac{Z_2}{1+\beta Z_1}.\label{sseq3}
\end{align}
Equations (\ref{sseq3}) lead to the following $n+3$ dimensional scalar system:
\begin{align}
& \frac{\mu}{K}=\frac{Z_2}{1+\beta Z_1},\,F^*\frac{Z_1}{1+\beta Z_1}=\displaystyle\sum_{i=1}^nZ_{i+2}\int_{x_1}^{x_2}\frac{k_i(x)}{a(x)}\,d x,\, F^*\frac{Z_2}{1+\beta Z_1}=\sum_{i=1}^nZ_{i+2}\int_{x_1}^{x_2}k_i(x)\,d x, \label{sseq1sep} \\
& F^*\frac{Z_3}{1+\beta Z_1}=\sum_{i=1}^nZ_{i+2}\int_{x_1}^{x_2}\bar{k}_1(x)\frac{k_i(x)}{a(x)}\,d x,\quad\cdots,\quad F^*\frac{Z_{n+2}}{1+\beta Z_1}=\sum_{i=1}^{n}Z_{i+2}\int_{x_1}^{x_2}\bar{k}_n(x)\frac{k_i(x)}{a(x)}\,d x. \label{sseq2sep}
\end{align}
Condition (\ref{sscond1}) for a separable kernel suggests that for small enough values of the parameter $\beta$ one should be able to establish existence (but not necessarily uniqueness) of a positive stationary state. In particular, for $\beta=0$
the first and third equation of (\ref{sseq1sep}) together yield

\begin{equation}
F^*=\sum_{i=1}^n \alpha_iZ_{i+2},\quad \alpha_i=\frac{K}{\mu}\int_{x_1}^{x_2}k_i(x)\,d x,\quad i=1,\cdots,n.\label{sseqsepF}
\end{equation}
With this, the $n$-dimensional nonlinear scalar system \ref{sseq2sep} can be cast in the form:
\begin{equation}
Z_3\sum_{i=1}^{n}\alpha_iZ_{i+2}=\sum_{i=1}^n\kappa_{1,i}Z_{i+2},\,\cdots,\, Z_{n+2}\sum_{i=1}^n\alpha_iZ_{i+2}=\sum_{i=1}^n\kappa_{n,i}Z_{i+2}, \label{sseq3sep}
\end{equation}
where
\begin{equation}
\kappa_{i,j}=\int_{x_1}^{x_2}\frac{\bar{k}_i(x)k_j(x)}{a(x)}\,d x,\quad i,j=1,\cdots,n.\label{kappaij}
\end{equation}

Once a non-negative solution of system (\ref{sseq3sep}) is found it can be substituted into equation (\ref{sseqsepF})
to determine $F^*$, which using the second equation of (\ref{sseq1sep}) determines a unique positive $Z_1$.

To establish existence of a positive solution of system (\ref{sseq3sep}) we utilise an idea which was employed for infinite dimensional problems recently, for example, in \cite{AF_2013, FGH_2010}. The key idea is to recast the non-linear problem (\ref{sseq3sep})
as an eigenvalue problem for a parameterised family of matrices. That is, we rewrite system (\ref{sseq3sep}) as
\begin{equation}
{\bf Z}=K_c\,{\bf Z},\quad {\bf Z}=(Z_3,\cdots,Z_{n+3})^T,\quad K_c(i,j)=\frac{\kappa_{i,j}}{c},\quad i,j=1,\cdots, n,\quad c>0,
\end{equation}
where
\begin{equation}
c=\sum_{i=1}^n\alpha_iZ_{i+2}.
\end{equation}

Note that for any $c>0$ the matrix $K_c$ is non-negative. It follows from Perron-Frobenius theory that the spectral radius $r(K_c)$  is an eigenvalue with a corresponding non-negative eigenvector. Also note that the function $c\to r(K_c)$ is continuous for $c\in (0,\infty)$. It follows for example from Gershgorin's Circle Theorem that as $c\to\infty$ we have $r(K_c)\to 0$. Hence if there exists a value $\bar{c}\in (0,\infty)$ such that $r(K_{\bar{c}})>1$, then there exists a $c_*$ such that $r(K_{c_*})=1$ and therefore $1$ is an eigenvalue with a corresponding non-negative eigenvector ${\bf Z^*}$. We then normalize this eigenvector such that it satisfies
\begin{equation}
c_*=\sum_{i=1}^n\alpha_i Z^*_{i+2}.
\end{equation}
Note that the existence of a $\bar{c}$ such that $r(K_{\bar{c}})>1$ holds depends on the $\kappa_{i,j}$ values, i.e. on the approximation of the kernel $k(x,y)$.
In particular if the kernel $k$ is such that for some $x$ value it is concentrated on the diagonal point $(x,x)$  then
again Gershgorin's Circle Theorem implies the existence of a $\bar{c}$ such that $r(K_{\bar{c}})>1$.

Also note that if the kernel $k$ is such (typically strictly positive) that it can be approximated with $\displaystyle\sum_{i=1}^n k_i(x)\bar{k}_i(y)$ such that $\kappa_{i,j}\ne 0$ for $i,j=1,\cdots,n$, then the matrix $K_c$ above is positive, and the Perron-Frobenius Theorem guarantees the existence of a strictly positive eigenvector ${\bf Z}$.
Also if the matrix $K_c$ is positive, then it is clear that for any $c<\displaystyle\min_{i,j}\left\{\kappa_{i,j}\right\}$ we have $r(K_c)>1$, hence it follows from the Intermediate Value Theorem that there exists a $c_*$ such that $r(K_{c_*})=1$. We summarize our findings in the following lemma.
\begin{lemma}
Assume that $k(x,y)=\displaystyle\sum_{i=1}^nk_i(x)\bar{k}_i(y)$, such that $\kappa_{i,j}>0$ for $i,j=1,\,\cdots,\, n$.
Then model (\ref{eq:Zoo})-(\ref{eq:Fish}) admits a strictly positive stationary state for $\beta=0$.
\end{lemma}

\noindent Next we rewrite system \eqref{sseq1sep}-\eqref{sseq2sep} as follows:
\begin{align}
& (1+\beta Z_1)\frac{\mu}{K}-Z_2=0, \label{sseq4sep} \\
& (1+\beta Z_1)\sum_{i=1}^n\bar{\alpha}_i Z_{i+2}-F^*Z_1=0,\quad (1+\beta Z_1)\sum_{i+1}^n\gamma_iZ_{i+2}-F^*Z_2=0,\label{sseq5sep} \\
& (1+\beta Z_1)\sum_{i=1}^n\kappa_{1,i}Z_{i+2}-F^*Z_3=0\,,\cdots,\,(1+\beta Z_1)\sum_{i+1}^n\kappa_{n,i}Z_{i+2}-F^*Z_{n+2}=0,\label{sseq6sep}
\end{align}
where
\begin{equation*}
\bar{\alpha}_i=\int_{x_1}^{x_2}\frac{k_i(x)}{a(x)}\,d x,\quad \gamma_i=\int_{x_1}^{x_2} k_i(x)\,d x,\quad i=1,\,\cdots,\, n.
\end{equation*}
Equations (\ref{sseq4sep})-(\ref{sseq6sep}) can be recast in the more economic form:
\begin{equation}
A(\beta,Z_1,Z_2,F^*,{\bf Z})={\bf 0}\,\,\in\mathbb{R}^{n+3},\label{Aopeq}
\end{equation}
where $A$ is well-defined (via the left hand-sides of equations (\ref{sseq4sep})-(\ref{sseq6sep})) on $\mathbb{R}^{n+4}$ and continuously differentiable. We would like to apply the Implicit Function Theorem for $A$
to show that if the equation $A(0,Z_1,Z_2,F^*,{\bf Z})={\bf 0}$ has a strictly positive solution then it also has a strictly positive solution for some small positive values of $\beta$. To this end we compute the Jacobian of $A$ at $(0,Z_1,Z_2,F^*,{\bf Z})$:
\begin{equation}
J=
\begin{pmatrix}
0 & -1 & 0 & 0 & \cdots & 0 \\
-F^* & 0 & -Z_1 & \bar{\alpha}_1 & \cdots & \bar{\alpha}_n  \\
0 & -F^* & -Z_2 & \gamma_1 & \cdots & \gamma_n \\
0 & 0 & -Z_3 & \kappa_{1,1} & \cdots & \kappa_{1,n} \\
\cdots & \cdots & \cdots & \cdots & \cdots & \cdots \\
0 & 0 & -Z_{n+2} & \kappa_{n,1} & \cdots & \kappa_{n,n}
\end{pmatrix}.
\end{equation}
Hence if the determinant of the Jacobian matrix evaluated at the a strictly positive stationary state is not zero then
by the Implicit Function Theorem a strictly positive stationary state also exists for small enough values of $\beta$.
Note that the value of the determinant of $J$ depends on the particular finite rank approximation of the kernel $k$.

\subsection{The general case}
We briefly discuss here how the steady state problem can be formulated in case of a general (i.e., non-separable) kernel $k$.
This case is challenging from the mathematical point of view since the integral operator describing the recruitment process is of infinite rank, in general. For a positive stationary state $(z^*(x),F^*)$ we define
\begin{equation}
\kappa^*(x)=\int_{x_1}^{x_2}k(x,y)z^*(y)\,dy,\,\, x\in[x_1,x_2],\quad Z^*=\int_{x_1}^{x_2}z^*(x)\,dx.
\end{equation}
With this notation the (positive) steady state problem can be formulated at least for sufficiently small values of $\beta$ (e.g.  for $\beta=0$) as follows
\begin{align}
\kappa^*(x)& =\frac{1+\beta Z^*}{F^*}\int_{x_1}^{x_2}\frac{k(x,y)}{a(y)}\kappa^*(y)\,dy,\quad x\in [x_1,x_2], \label{eqop1} \\
F^*& =\frac{K}{\mu}\int_{x_1}^{x_2}\kappa^*(x)\,dx,\label{eqop2} \\
Z^*& =\frac{\int_{x_1}^{x_2}\frac{\kappa^*(x)}{a(x)}\,dx}{F^*-\beta\int_{x_1}^{x_2}\frac{\kappa^*(x)}{a(x)}\,dx}. \label{eqop3}
\end{align}
Problem \eqref{eqop1}-\eqref{eqop3} can be considered as an eigenvalue problem for a bounded operator $\mathcal{O}\,:\,\mathcal{X}\times\mathbb{R}^2\to\mathcal{X}\times\mathbb{R}^2$, where $\mathcal{O}$ is defined via the right hand side of equations \eqref{eqop1}-\eqref{eqop3}, and $\mathcal{X}$ is an appropriately choosen Banach space, for example $L^1(x_1,x_2)$. More precisely, if $\mathcal{O}$ has eigenvalue $1$ with a corresponding strictly positive eigenvector $(\kappa^*,F^*,Z^*)^t$ then the system admits a positive steady state with
\begin{equation}\label{ssformgeneral}
z^*(x)=\frac{1+\beta Z^*}{F^*}\frac{\kappa^*(x)}{a(x)},\quad x\in [x_1,x_2].
\end{equation}
Note that for $\beta=0$ the operator $\mathcal{O}$ is positive,  and the eigenvalue problem for $\mathcal{O}$ can be analysed using similar arguments as in Section 2.2 above for infinite dimensional problems see e.g. \cite{AF_2013,KR_1964}.
In particular, for the special case of $\beta=0$ and $a\equiv A$ we have the following result.
\begin{theorem}\label{ssgeneral}
In the case of $\beta=0$ and $a\equiv A$ model \eqref{eq:Zoo}-\eqref{eq:Fish} admits a unique positive stationary state.
\end{theorem}
The proof of Theorem \ref{ssgeneral} is included in Appendix A. It is worth pointing out that, unlike in the previous sections, we are able to prove uniqueness of the coexistence steady state. To prove the existence of a positive steady state in the most general case, i.e. without the assumptions $\beta=0$ and $a\equiv A$,
we could apply results from the forthcoming paper \cite{CF_2014}, but due to the technical difficulties involved this is outside the scope of the present paper.

%%%%%%%%%%%%%%%%%%%%%%%%%%%%%%%%%%%%%%%%%%%%%%%%%%%%%%%%%%%%%%%%%%%%%%%%%%%%%%%%%%%%%%%%%%%%%%%%%%%%%%%%%%%%%%%%%%%%%%%%%%%%%%%%%%%%%%%%%%%%%%%%%%%%%%%%%%%%%%%%%%%%%%%%%%%%%%%%%%%%%%%%%%%%%%%%%%%%%%%%%%%%%%%

\section{Stability Analysis}

In this section we shall address the local stability of the coexistence stationary state of the model. To this end we linearise the model around the steady state, and analyse the arising eigenvalue problems.
Our model is semi-linear, and so we may invoke Theorem 5.1.1 and Theorem 5.1.3 from \cite{Henry_1981} to justify that the stability results obtained in the section are valid. We linearise equations (\ref{eq:Zoo})-(\ref{eq:Fish}) in the vicinity of the stationary state $(F^*,z^*)$ and consider perturbations $w(x,t)=z(x,t)-z^*(x)$ and $G(t)=F(t)-F^*$ of this state.
We obtain
\begin{align}
\frac{\partial}{\partial t}w(x,t)&=\int_{x_1}^{x_2}k(x,y)w(y,t) dy \nonumber \\
& - \left(\frac{a(x)z^*(x)}{1+\beta Z^*}G(t)+\frac{F^*a(x)}{1+\beta Z^*}w(x,t)-\frac{F^*a(x)z^*(x)\beta}{(1+\beta Z^*)^2}W(t)  \right)\label{eq:Stability_general1} \\
\frac{\ d}{\ d t}G(t)&=-\mu G(t) \nonumber \\
&+K\int_{x_1}^{x_2}\left(\frac{a(x)z^*(x)}{1+\beta Z^*}G(t)+\frac{F^*a(x)}{1+\beta Z^*}w(x,t)-\frac{F^*a(x)z^*(x)\beta}{(1+\beta Z^*)^2}W(t)\right) d x,\label{eq:Stability_general2}
\end{align}
where $Z^*=\int_{x_1}^{x_2}z^*(x)\,\ d x$ and $W(t)=\int_{x_1}^{x_2}w(x,t)\,\ d x$.

We note that the linear problem above is governed by an analytic semigroup, hence the spectrum may contain
only eigenvalues of finite multiplicity. To determine the possible eigenvalues $\lambda$ we look for solutions
of the linearised equations in the standard form: $w(x,t)=\exp(\lambda t)w(x)$, and $G(t)=\exp(\lambda t)G$. This ansatz leads to the following eigenvalue problem (with the assumption of $k>0$):
\begin{align}
0&=w(x)\left(\lambda+\frac{F^*a(x)}{1+\beta Z^*}\right)-\int_{x_1}^{x_2}k(x,y)w(y)d y-W\frac{F^*\beta a(x)z^*(x)}{(1+\beta  Z^*)^2}+G\frac{a(x)z^*(x)}{1+\beta Z^*}, \label{eq:Stability_1} \\
0&=-K\int_{x_1}^{x_2}\frac{F^*a(x)}{1+\beta Z^*}w(x) dx+W\frac{F^*\beta\mu}{1+\beta Z^*}+G\lambda.\label{eq:Stability_2}
\end{align}

 As we can see the eigenvalue problem (\ref{eq:Stability_1})-(\ref{eq:Stability_2}) is rather complicated, in general, since it contains integral equations. Nevertheless in the following subsections we discuss some interesting special cases when the eigenvalue problem becomes tractable, and we can deduce analytical stability or instability results.

\subsection{Constant kernel $k(x,y)\equiv C$.}
The eigenvalue problem for the case of a constant kernel $k(x,y)\equiv C$ reads:
\begin{eqnarray}
0&= & w(x)\left(\lambda+\frac{F^*a(x)}{1+\beta Z^*}\right)+W\left(-C-\frac{F^*\beta a(x)z^*(x)}{(1+\beta Z^*)^2}\right)+G\frac{a(x)z^*(x)}{1+\beta Z^*}, \label{evalue1_1} \\
0&= & -K\int_{x_1}^{x_2}\frac{F^*a(x)}{1+\beta Z^*}w(x)dx+W\frac{F^*\beta\mu}{1+\beta Z^*}+G\lambda.\label{evalue1_2}
\end{eqnarray}
We integrate equation (\ref{evalue1_1}), then multiply it by $K$ and add it to equation (\ref{evalue1_2}) to obtain:
\begin{eqnarray}\label{evalue3_1}
0=W(\lambda K-CK(x_2-x_1))+G(\lambda+\mu).
\end{eqnarray}
Since any non-trivial eigenvector is determined up to a constant multiplier, we may assume that the eigenvector $(w,G)^T$ is such that $W=1$. Note that since we only want show the existence of a positive eigenvalue with a corresponding non-trivial eigenvector (hence we are not characterising the whole point spectrum) we do not need to worry about possible eigenvalues with eigenvectors for which $W=0$. With $W=1$ from (\ref{evalue3_1}) we obtain (for $\lambda\ne-\mu$):
\begin{eqnarray}
\label{Gvalue}
G=\frac{CK(x_2-x_1)-\lambda K}{\lambda+\mu}.
\end{eqnarray}
Using (\ref{Gvalue}) we obtain from (\ref{evalue1_1}) $\left(\text{for}\,\,\lambda\ne -\frac{F^*a(x)}{1+\beta Z^*}\right)$
\begin{eqnarray}
\label{Main}
w(x)=\frac{C+\frac{F^*\beta a(x)z^*(x)}{(1+\beta Z^*)^2}+\frac{\lambda K-CK(x_2-x_1)}{\lambda+\mu}\frac{a(x)z^*(x)}{1+\beta Z^*}}{\lambda+\frac{F^*a(x)}{1+\beta Z^*}}.
\end{eqnarray}

The equation for the eigenvalues can be derived by integrating (\ref{Main}) over $[x_1,x_2]$ which should give unity (according to our choice of $W$):
\begin{eqnarray}
\label{Main1}
\int_{x_1}^{x_2} w(x) dx=\int_{x_1}^{x_2} \frac{C+\frac{F^*\beta a(x)z^*(x)}{(1+\beta Z^*)^2}+\frac{\lambda K-CK(x_2-x_1)}{\lambda+\mu}\frac{a(x)z^*(x)}{1+\beta Z^*}}{\lambda+\frac{F^*a(x)}{1+\beta Z^*}}dx=W=1.
\end{eqnarray}
We can simplify the expression above by taking into account the explicit formulae for the stationary densities $(z^*,F^*)$ in the case of a constant kernel $k$. Without losing generality we can assume that $k=C=\frac{R}{x_2-x_1}=\frac{R}{h}$. In this case, $R$ has the meaning of the growth rate of the whole prey population (for a constant kernel the growth rate does not depend on the particular cohort). Using (\ref{scalar1})-(\ref{scalar4}), we have the stationary densities of species:

\begin{eqnarray}
Z^* &=& \frac{\mu \int_{x_1}^{x_2}\frac{dx}{a(x)}}{Kh-\mu \beta \int_{x_1}^{x_2}\frac{dx}{a(x)}}, \label{biomassZ} \\
z^*(x) &=& \frac{\mu}{\left(Kh-\mu \beta \int_{x_1}^{x_2}\frac{dx}{a(x)}\right)a(x)}, \label{densityz} \\
F^* &=& \frac{K R \int_{x_1}^{x_2}\frac{dx}{a(x)}}{Kh-\mu \beta \int_{x_1}^{x_2}\frac{dx}{a(x)}}.\label{densF}
\end{eqnarray}
We substitute the stationary values (\ref{biomassZ})-(\ref{densF}) into \eqref{Main1}, and after some simplification we obtain:
\begin{eqnarray}
\label{Main2}
\int_{x_1}^{x_2} \frac{\frac{R}{h}+\beta\mu\frac{R}{h^2 K} I_1+\frac{\lambda-R}{\lambda+\mu}\frac{\mu}{hK}}{\lambda+a(x)\frac{R}{h}
I_1}dx &=&1,
\end{eqnarray}
where $I_1=\int_{x_1}^{x_2}\frac{dx}{a(x)}$. We summarize our finding in the following theorem.
 \begin{theorem}\label{k-const}
The positive stationary state $(z^*,F^*)$ of (\ref{eq:Zoo})-(\ref{eq:Fish}) with a constant kernel $k\equiv R/h$ is asymptotically stable if all of the solutions of (\ref{Main2}) have negative real parts, and it is unstable if there exists at least one solution of \eqref{Main2} with positive real part.
\end{theorem}

By using a particular parametrisation of $a(x)$ the characteristic equation \eqref{Main2} becomes analytically tractable and one can deduce sufficient conditions for the stability of the stationary state. A particularly important case, however, is when the saturation in the predation rate $\beta$ is small ($\beta\ll1$). In this case we can set $\beta=0$,  and the characteristic  equation (\ref{Main2}) takes the following simple form:
\begin{eqnarray}
\label{Main3}
\int_{x_1}^{x_2} \left(\frac{\frac{R}{h}+\frac{\lambda-R}{\lambda+\mu}\frac{\mu}{hK}}{\lambda+a(x)\frac{R}{h}
I_1}\right)dx &=&1.
\end{eqnarray}
Utilising this equation, we can analytically explore stability properties of the stationary state for an arbitrary function $a(x)$ in the case when the trait interval is small, i.e. when $h=x_2-x_1\ll1$. We can use a first order approximation of $a(x)$ given by $a(x)=a_0+a_1(x-x_1)$. We substitute this approximation of $a(x)$ into (\ref{Main3}), and after integration we obtain:
\begin{eqnarray}
\label{linear_a}
\left( \frac{R}{h}+\frac{(\lambda-R)}{\lambda+\mu}\frac{\mu}{h K}  \right)  \ln\left[\frac{a_1R I_1 x_2+\lambda h+R (a_0-a_1 x_1)I_1}{\lambda h+a_0 R I_1}\right] \frac{h}{a_1 R I_1}&=&1.
\end{eqnarray}
We compute the Taylor expansion of the left hand side of (\ref{linear_a}) up to the second order of $h$ and multiply it by $(\lambda+\mu)(\lambda+R)^3$ to obtain a fourth order polynomial function in terms of $\lambda$. We can do it since we look for only $\lambda$ with $\Re(\lambda)>0$
\begin{eqnarray}
\label{linear_a2}
\lambda^4+A_3\lambda^3+A_2 \lambda^2+A_1\lambda+A_0 &=&0,
\end{eqnarray}
where the polynomial coefficients $A_i$ are given by:
\begin{eqnarray*}
\label{P}
A_0 &=& R^3\mu,\, A_1  = 2R^2\mu,A_2 = \frac{R(12a_0^2+a_1^2h^2)(R+\mu)}{12a_0 ^2},\, A_3 = 2R.
\end{eqnarray*}
To determine stability of the polynomial (\ref{linear_a2}) we can utilise the Routh-Hurwitz stability criterion, see e.g. in  \cite{Rahman2002}:

\begin{eqnarray*}
\label{P1}
A_3 A_2 > A_1, \, A_3 A_2 A_1>A_2^2-A_3^2 A_0, \, A_i>0.
\end{eqnarray*}
Verification of the above conditions show that we have always $A_i>0$ and
\begin{eqnarray*}
\label{P2}
A_3 A_2 - A_1 = \frac{R^2(12a_0^2R+Ra_1^2h^2+a_1^2 h^2\mu)}{6 a_0^2}>0, \\
A_3 A_2 A_1-A_2^2-A_3^2 A_0 = R^4a_1^2h^2 \mu \frac{R+\mu}{3 a_0^2}>0.
\end{eqnarray*}
thus the system is locally stable for small $h$. Interestingly, the actual sign of $a_1$, i.e. the gradient of the dependence of vulnerability to predation on the life history trait $x$, does not affect the stability result: the system becomes stabilised for any $a_1 \ne 0$. Note also that keeping only the linear part in terms of $h$ in the expansion of (\ref{linear_a}) would be misleading since some roots $\lambda$ of the polynomial will be purely complex resulting in a degenerated case: small variations of this equation (e.g. considering $1\gg \beta>0$) will perturb those solutions resulting in the appearance of the non-zero real parts.

In the case when the interval $[x_1,x_2]$ is not small and we cannot always assume that $h\ll1$ (but we can still neglect the saturation in the predation rate, i.e. set  $\beta=0$) we need to solve the full equation (\ref{Main3}), which in general is a difficult task since $\lambda$ has both real and imaginary parts. However, we know that for a sufficiently small $h$ the stationary state is always stable and due to the  continuous dependence of the model dynamics on the parameters  (by invoking Theorem 3.4.1 and Theorem 3.4.4 from \cite{Henry_1981}), the stability loss (if any) of the stationary state for a larger value of $h$ should take place via a Hopf bifurcation with $\Re(\lambda)=0$. To find the bifurcation point we substitute $\lambda=i\omega$ into (\ref{Main3}) and separate the real and imaginary part of this equation. After some rearrangement we obtain the following system of equations:
\begin{eqnarray}
\label{Im_1}
 \int_{x_1}^{x_2}\frac{\omega^2}{\omega^2+\left[a(x)\frac{R}{h}I_1\right]^2}dx &=&\frac{h \mu}{R+\mu},\\
 \label{Re_1}
  \int_{x_1}^{x_2}\frac{a(x)\frac{R}{h}I_1}{\omega^2+\left[a(x)\frac{R}{h}I_1\right]^2}dx &=&\frac{h}{R+\mu},
\end{eqnarray}
where the integral $I_1$ is defined earlier.

One can further explore the possibility of solving equations (\ref{Im_1})-(\ref{Re_1}) for a particular function $a(x)$. In our study we have considered the generic linear and parabolic functions given by $a(x)=a_0+a_1(x-x_1)$ and $a(x)=a_0+a_1(x-x_1)+a_1(x-x_1)^2$. Note that for those functions it is easy to calculate the above integrals analytically. In each case we investigated the possibility of solving the system (\ref{Im_1})-(\ref{Re_1}), i.e. to find $\omega$ satisfying both equations. Such $\omega$ will correspond to a Hopf bifurcation point. Our results show that for the given functions there is no solution for any combination of the other model parameters. We do not show here the results for the sake of brevity. Thus, we can conclude that for linear and parabolic parameterisations of $a(x)$  the stationary state $(z^*,F^*)$ is always locally asymptotically stable for any trait interval $[x_1,x_2]$, whenever it exists. It remains to address whether or not this stability result can be extended to the case of an arbitrary function $a(x)$ and $\beta=0$.

On the other hand, in the case of a sufficiently large saturation in predation (i.e. when we cannot neglect $\beta$) the stationary state $(z^*,F^*)$ is always unstable, at least for small values of $h$. This can be shown directly by expanding equation (\ref{Main2}) into Taylor series and keeping only the linear part with respect to $h$.
\begin{eqnarray}
\label{beta_1}
 \mu R(\mu \beta-a_0 K) +\beta \mu R \lambda-a_0 \lambda^2.
\end{eqnarray}
The eigenvalues of this equation have positive real parts since the coefficients of the polynomial have different signs ($a_0>0$), thus the stationary state is unstable. Since for small $\beta$ the stationary state is stable (at least for small $h$, see above), but it becomes unstable for larger $\beta$, there should be a critical value of $\beta=\beta(h)$, where the stability switch occurs. The function $\beta(h)$ defines a Hopf bifurcation curve, which can be obtained numerically from (\ref{Main2}) by substituting $\lambda=i\omega$. The equation for the bifurcation curve is defined by the following system, which we obtained by equating to zero both the imaginary and the real parts of equation (\ref{Main2}):
\begin{eqnarray}
 & \int_{x_1}^{x_2}\frac{dx}{\omega^2+\left[a(x)\frac{R}{h}I_1\right]^2}dx \hspace{45mm} \nonumber  \\
 =&\frac{\mu h^3 K^2 (R+\mu) }{h^2K^2 \omega^2(R+\mu)^2+R\mu \beta I_1 \left[\mu R \beta I_1 (\mu^2+\omega^2)+2hK\omega^2 (R+\mu) \right]}, \label{Im_2}\\
  & \int_{x_1}^{x_2}\frac{a(x)\frac{R}{h}I_1 dx}{\omega^2+\left[a(x)\frac{R}{h}I_1\right]^2}dx \hspace{45mm} \nonumber \\
  =&\frac{\left[ hK\omega^2(R+\mu)+\mu R \beta I_1 (\mu^2 +\omega^2) \right]}{h^2K^2 \omega^2(R+\mu)^2+R\mu \beta I_1 \left[\mu R \beta I_1 (\mu^2+\omega^2)+2hK\omega^2 (R+\mu) \right]}, \label{Re_2}
\end{eqnarray}
with $\omega$ being a positive real value. By choosing a particular function $a(x)$ and solving (\ref{Im_2})-(\ref{Re_2}) one can construct the Hopf bifurcation curve. Note that formally one need to check the standard requirement of , see e.g.  \cite{Perko}, i.e. that $d \Re(\lambda)/d\alpha \ne 0$ holds, with $\alpha$ being a bifurcation parameter (we can verify this condition numerically). We emphasize that in (\ref{Im_2})-(\ref{Re_2}) we do not use the assumption that $h\ll1$.

As an illustrative example, we constructed a family of Hopf bifurcation curves numerically for a linear vulnerability function $a(x)=a_0+a_1(x-x_1)$ in the $(\beta, a_1)$ plane, which is shown in Figure 1. Note that for the given $a(x)$, the integrals in  (\ref{Im_2})-(\ref{Re_2}) can be easily calculated analytically and the system becomes a system of transcendental equations. For each curve in the figure, the region corresponding to the stable stationary state is located on the left-hand side of the curve. The different curves correspond to different lengths of the  interval $[x_1, x_2]$. One can see that stability loss due to an increase in $\beta$ can be compensated for either by increasing the range of $x$, or by increasing the absolute value of the gradient of $a(x)$. This is in agreement with earlier results in \cite{Morozov2013}, which were obtained via direct simulation of the model equations.
\begin{figure}[ht!]
\centering
\includegraphics[width=12cm]{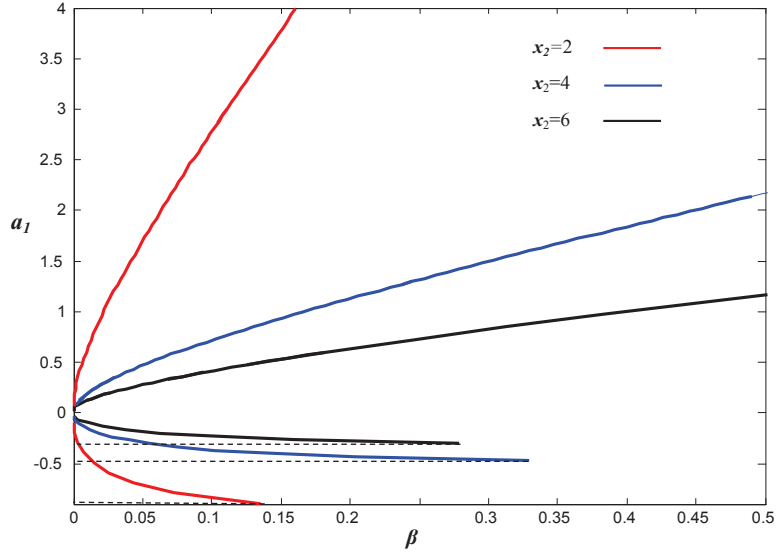}
 \caption{Examples of Hopf bifurcation curves for model (\ref{eq:Zoo})-(\ref{eq:Fish}) with a constant kernel $k\equiv R/h$ and the linear vulnerability $a(x)=a_0+a_1(x-x_1)$ constructed using equations (\ref{Im_2})-(\ref{Re_2}). Different curves correspond to different values of $x_2$. The other model parameters are: $K=1, \mu=0.1, a_0=1, R=2$. For a given $x_2$ the stability region is located on the left-hand side of each curve. The dashed horizontal lines show the lower limits of $a_1$ beyond which $a(x)$ becomes negative.}
\label{fig1}
\end{figure}
Note that equations (\ref{Im_2})-(\ref{Re_2}) only account for the fact that the stationary state loses it stability, however, without providing information about which type of Hopf bifurcation actually occurs: a subcritical or a supercritical one. The knowledge about the type of Hopf bifurcation can be of importance since after a supercritical bifurcation the trajectories remain bounded close to the equilibrium whereas they can become unbounded (no population persistence, see \cite{Morozov2013}) in case of subcritical Hopf bifurcation. Numerical simulations carried out with particular parametrisations of $a(x)$ show that we can have both types of bifurcations (depending on the model parameters). However, revealing a particular type of Hopf bifurcation would require a rather lengthy expressions and should be a matter of a separate study.

\subsection{The kernel is a function of $y$ only: $k(x,y)\equiv k(y)$}

The main technique of stability analysis is similar to the one implemented in the case of a constant kernel. Equations (\ref{eq:Stability_1})-(\ref{eq:Stability_2})  now take the following form:
\begin{align}
0&=w(x)\left(\lambda+\frac{F^*a(x)}{1+\beta Z^*}\right)-\int_{x_1}^{x_2}w(x)k(x) dx-W \frac{F^*\beta a(x)z^*(x)}{(1+\beta Z^*)^2}+G\frac{a(x)z^*(x)}{1+\beta Z^*}, \label{evalue1} \\
0&=-K\int_{x_1}^{x_2}\frac{F^*a(x)}{1+\beta Z^*}w(x)dx+W\frac{F^*\beta\mu}{1+\beta Z^*}+G\lambda.\label{evalue2}
\end{align}
We integrate equation (\ref{evalue1}), multiply it by $K$ and add it to (\ref{evalue2}) to obtain:
\begin{eqnarray}\label{evalue3}
0=W\lambda K-K B (x_2-x_1)+G(\lambda+\mu),
\end{eqnarray}
where $B=\int_{x_1}^{x_2}w(x)k(x) dx$.

We can again assume that $W=1$ (as before) and we do not need to worry about possible eigenvalues with eigenvectors for which $W=0$. We obtain for $\lambda\ne-\mu$:
\begin{eqnarray}
\label{Gvalue_1}
G=\frac{B K(x_2-x_1)-\lambda K}{\lambda+\mu}.
\end{eqnarray}
Using (\ref{Gvalue_1}) we obtain from (\ref{evalue1}) $\left(\text{for}\,\,\lambda\ne -\frac{F^*a(x)}{1+\beta Z^*}\right)$
\begin{eqnarray}
\label{Main_4}
w(x)=\frac{B+\frac{F^*\beta a(x)z^*(x)}{(1+\beta Z^*)^2}+\frac{\lambda K-B K(x_2-x_1)}{\lambda+\mu}\frac{a(x)z^*(x)}{1+\beta Z^*}}{\lambda+\frac{F^*a(x)}{1+\beta Z^*}}.
\end{eqnarray}
We multiply equation (\ref{Main_4}) by $k(x)$ and integrate it over $[x_1,x_2]$, this should give us the value of $B$:
\begin{eqnarray}
\label{Main_5}
\int_{x_1}^{x_2} w(x)k(x) dx=\int_{x_1}^{x_2} \frac{B+\frac{F^*\beta a(x)z^*(x)}{(1+\beta Z^*)^2}+\frac{\lambda K-BK(x_2-x_1)}{\lambda+\mu}\frac{a(x)z^*(x)}{1+\beta Z^*}}{\lambda+\frac{F^*a(x)}{1+\beta Z^*}}dx=B.
\end{eqnarray}
We can simplify equation (\ref{Main_4}) by using the explicit expressions for the stationary densities $(z^*,F^*)$ given by equations (\ref{scalar1})-(\ref{scalar4}). The stationary values of $Z^*$ and $z^*$ will be the same as in (\ref{biomassZ}) and (\ref{densityz}). For $F^*$ we have:
\begin{eqnarray}
F^* &=& \frac{K \int_{x_1}^{x_2}\frac{k(x) dx}{a(x)}dx}{Kh-\mu \beta \int_{x_1}^{x_2}\frac{dx}{a(x)}dx}.\label{densityF}
\end{eqnarray}
We substitute these stationary values into (\ref{Main_5}), and after some simplification we obtain:
\begin{eqnarray}
\int_{x_1}^{x_2} \frac{B+\frac{\beta \mu}{K h} I_2+\frac{\lambda K-B K(x_2-x_1)}{\lambda+\mu}\frac{\mu}{h K}}{\lambda+a(x) I_2} \, k(x) dx&=&B,
\end{eqnarray}
where $I_2=\displaystyle\int_{x_1}^{x_2}\frac{k(x)}{a(x)}dx$. Then we obtain the solution $B$ as follows:
\begin{eqnarray}
B=- I_3 \mu \, \frac{I_2 \beta(\lambda+\mu)+\lambda K}{(I_3\lambda-\lambda-\mu)Kh},
\end{eqnarray}
where $I_3=\displaystyle\int_{x_1}^{x_2}\frac{k(x)}{\lambda+a(x) I_2}dx$.
Finally, we substitute $B$ into the equation for $w(x)$ and integrate over $[x_1, x_2]$, this should give us 1 since we set $W=1$. After some simplification, we have:
 \begin{eqnarray}
 \label{Main_7}
-\mu \, \frac{I_2 \beta(\lambda+\mu)+\lambda K}{\left(\lambda \displaystyle\int_{x_1}^{x_2}\frac{k(x)}{\lambda+a(x) I_2}dx-\lambda-\mu\right)Kh} \int_{x_1}^{x_2} \frac{1}{\lambda+a(x) I_2} \, dx&=&1.
\end{eqnarray}
We summarize our findings in the following theorem.
\begin{theorem}\label{k(x)}
The positive stationary state $(z^*,F^*)$ of (\ref{eq:Zoo})-(\ref{eq:Fish}) with a kernel $k\equiv k(y)$ is stable if the eigenvalues of characteristic equation (\ref{Main_7}) have negative real parts, and it is unstable if there is at least one  eigenvalue with positive real parts.
\end{theorem}

Note that it is possible (using the same method as in the previous section) to prove the stability of the system in the case when $\beta$ and the length of the trait interval $h$ are small. Again, this stability does not depend on the sign of the gradient of $a(x)$. On the contrary, in the case when $\beta$ is not small, but $h$ is small, we can prove that the stationary state is unstable. Thus again, there should exist a Hopf bifurcation curve which separates the regions of stability and instability. For any particular function $a(x)$, construction of such a curve should be based on equation (\ref{Main_7}), where we replace $\lambda$ with $i \omega$.

Finally, we can consider the situation where the vulnerability to predation is constant: $a(x)\equiv A$. In this case the characteristic equation reduces to:
 \begin{eqnarray}
 \label{a-const}
-\mu \, \frac{I_2 \beta(\lambda+\mu)+\lambda K}{(\lambda \frac{A I_2}{\lambda+A I_2}-\lambda-\mu)K(\lambda+A I_2)} \, &=&1,
\end{eqnarray}
which can be re-written as an equivalent equation (we are interested to find only $\lambda$ with $\Re (\lambda)>0$)
 \begin{eqnarray}
 \label{a-const}
\lambda^2 -I_2 \beta \mu \lambda +I_2 \mu (AK-\beta \mu)&=&0.
\end{eqnarray}
It is easy to see that for $\beta>0$ we shall always have eigenvalues with positive real part, thus the equilibrium will be  unstable for any $k\equiv k(y)$.
%%%%%%%%%%%%%%%%%%%%%%%%%%%%%%%%%%%%%%%%%%%%%%%%%%%%%%%%%%%%%%%%%%%%%%%%%%%%%%%%%%%%%%%%%%%%%%%%%%%%%%%%%%%%%%%%%%%%%%%%%%%%%%%%%%%%%%%%%%%%%%%%%%%%%%%%%%%%%%%%%%%%%%%%%%%%%%%%%%%%%%%%%%%%%%%%%%%%%%%%%%%%%%%

\subsection{Separable kernel of a finite rank}
We can extend the previous techniques of stability analysis to deal with a more general case, where the kernel is given by (\ref{finite_rank}). We assume that we have already calculated the values of $Z_i$ and $F^*$ (see section 2.2 for details). This can be done, for instance, by using numerical methods. Equations (\ref{eq:Stability_1})-(\ref{eq:Stability_2}) now read:
\begin{align}
0&=w(x)\left(\lambda+\frac{F^*a(x)}{1+\beta Z^*}\right)-k_i(x)\sum_{i=1}^n \int_{x_1}^{x_2}w(x)\bar{k}_i(x) dx-W \frac{F^*\beta a(x)z^*(x)}{(1+\beta Z^*)^2}+G\frac{a(x)z^*(x)}{1+\beta Z^*}, \label{evalue_f1} \\
0&=-K\int_{x_1}^{x_2}\frac{F^*a(x)}{1+\beta Z^*}w(x)dx+W\frac{F^*\beta\mu}{1+\beta Z^*}+G\lambda.\label{evalue_f2}
\end{align}
We introduce the notation $B_i=\displaystyle\int_{x_1}^{x_2}w(x)\bar{k}_i(x) dx$, multiply equation (\ref{evalue_f1}) by $K$, integrate this equation and add it to equation (\ref{evalue_f2}). We obtain:
\begin{eqnarray}\label{evalue3}
0=W\lambda K-K\sum_{i=1}^n B_i I_{0i}  +G(\lambda+\mu),
\end{eqnarray}
where $I_{0i}=\displaystyle\int_{x_1}^{x_2} k_i(x) dx$. We assume again that $W=1$, and express $G$ as follows:
\begin{eqnarray}
\label{Gvalue_f}
G=\frac{K\displaystyle\sum_{i=1}^n B_i I_{0i}-\lambda K}{\lambda+\mu}.
\end{eqnarray}
Using equation (\ref{Gvalue_f}) we obtain from (\ref{evalue_f1}) with $\frac{a(x)z(x)}{(1+\beta Z_1)}$ given by (\ref{sseq1sep})
\begin{eqnarray}
\label{Main_f}
w(x)=\frac{\displaystyle\sum_{i=1}^n B_i k_i(x) +\left[ \frac{\beta }{(1+\beta Z^*)}+\frac{\lambda-\displaystyle\sum_{i=1}^n B_i I_{0i}}{\lambda+\mu}\frac{K}{F^*}\right] \displaystyle\sum_{i=1}^n Z_{i+2}k_i(x)}{\lambda+\frac{F^*a(x)}{1+\beta Z^*}}.
\end{eqnarray}
To find the coefficients $B_i$ we multiply equation (\ref{Main_f}) by $\bar{k}_i(x)$ and integrate it over $[x_1, x_2]$ for $i=1,\cdots,n$. We obtain the following system of linear equations:
\begin{eqnarray}
\label{Main_ff}
B_j=\sum_{i=1}^n B_i I_{ij}+ \left[ \frac{\beta }{(1+\beta Z^*)}+\frac{\lambda-\sum_{i=1}^n B_i I_{0i}}{\lambda+\mu}\frac{K}{F^*} \right]\sum_{i=1}^n Z_{i+2} I_{ij},
\end{eqnarray}
where
\begin{equation*}
I_{ij}=\int_{x_1}^{x_2}\frac{k_i(x)\bar{k}_j(x)}{\lambda+\frac{F^*a(x)}{1+\beta Z^*}}dx.
\end{equation*}
After re-arrangement, system (\ref{Main_ff}) reads:
\begin{align}
-\frac{\beta }{(1+\beta Z^*)}-\frac{\lambda K}{\lambda+\mu} \sum_{i=1}^n Z_{i+2} I_{ij} &=& \sum_{i=1}^n B_i \left[ I_{ij}-\frac{I_{0i}}{\lambda+\mu}\frac{K}{F^*} \sum_{i=1}^n Z_{i+2} I_{ij} -\delta_{ij}\right],\label{Main_fff}
\end{align}
with $\delta_{ii}=1$ and $\delta_{ij}=0$, for $i\ne j $. We summarize our findings in the following theorem.

\begin{theorem}
Let us assume that system (\ref{sseq1sep})-(\ref{sseq2sep}) has a unique solution $(Z_i,F^*)$ and the determinant of the linear system (\ref{Main_fff}) is non-zero. Then the characteristic equation of the linearised system (\ref{eq:Stability_general1})- (\ref{eq:Stability_general2}) for a separable kernel (\ref{finite_rank}) is given by
\begin{align}
& \displaystyle\int_{x_1}^{x_2}\left[\left(\displaystyle\sum_{i=1}^n B_i k_i(x) +\left[ \frac{\beta }{(1+\beta Z^*)}+(\lambda+\mu)\left(\lambda-\displaystyle\sum_{i=1}^n B_i I_{0i}\frac{K}{F^*}\right)\right] \sum_{i=1}^n Z_{i+2}k_i(x)\right)\right. \nonumber \\
& \left.\hspace{45mm}\left(\lambda+\frac{F^*a(x)}{1+\beta Z^*}\right)^{-1}\right]dx=1,\label{Main_fm}
\end{align}
where $B_i=B_i(\lambda)$ are the solutions of equations (\ref{Main_fff}).
\end{theorem}
In particular, the Hopf bifurcation curve can be found by substituting $\lambda= i \omega$ into equation (\ref{Main_fm}) and solving the system of two equations for the imaginary and the real parts in a similar way as in (\ref{Im_2})-(\ref{Re_2}).

Although Theorem 3.3 allows a complete characterisation of the eigenvalues, its practical implementation might require some advanced computational methods which can be as much complicated as solving the original differential equations. However, in a particular case, when the kernel is separable, i.e.  $k=k_1(x)\bar{k}_1(y)$, the characteristic equation becomes much simpler. We can use the explicit expressions for the stationary densities of species given in section 2.1. After substituting the densities into (\ref{Main_fff}) we obtain the scalar equation for $B$, which we substitute into (\ref{Main_fm}) to obtain
 \begin{eqnarray}
 \label{Main_71}
-\mu \, \frac{I_3 \beta(\lambda+\mu)+\lambda K}{\left(\lambda \displaystyle\int_{x_1}^{x_2}\frac{k_1(x)\bar{k}_1(y)}{\lambda+a(x) I_2} d x-\lambda-\mu\right)K h I_1} \displaystyle\int_{x_1}^{x_2} \frac{1}{\lambda+a(x) I_2} \, dx&=&1,
\end{eqnarray}
where
\begin{equation*}
I_1=\displaystyle\int_{x_1}^{x_2} k_1(x)dx,\quad I_2=\displaystyle\int_{x_1}^{x_2} \frac{k_1(x)}{a(x)}dx, \quad I_3=\displaystyle\int_{x_1}^{x_2} \frac{k_1(x)\bar{k}_1(y)}{a(x)}dx.
\end{equation*}
Interestingly, in the particular case $a(x)\equiv A$, we can easily prove based on (\ref{Main_71}) that the stationary state is always unstable. We also managed to prove instability for the case $a(x)\equiv A$ when $k=k_1(x)\bar{k}_1(y)+k_2(x)\bar{k}_2(y)$, (the result is not shown here due to a rather lengthy proof); however, for the general case with an arbitrary number of terms in the kernel $k$, proving (in)stability remains an open problem. Note that in the next section, we partially address this question by providing a sufficient condition for instability for the case of $a(x)\equiv A$.

\subsection{Constant $a(x)\equiv A$}
In the previous section we partially addressed the question of stability of the stationary state in the case $a(x)\equiv A$ for a simple separable kernel and found that the state is always unstable. In this subsection, we shall consider a more general situation for an arbitrary positive kernel $k=k(x,y)$. The case of $a(x)\equiv A$ is of particular practical interest since it is important to know whether or not stability would be possible in the case of structuring in the prey population in terms of growth rate only and without the selectivity in predation according to the life trait $x$.

In case of $a(x)\equiv A$ the eigenvalue problem \eqref{eq:Stability_1}-\eqref{eq:Stability_2} reads:
\begin{align}
0&=G\lambda+W\frac{KF^*A}{1+\beta Z^*}\left(\frac{\beta Z^*}{1+\beta Z^*}-1\right), \label{eigva1} \\
0&=G\frac{Az^*(x)}{1+\beta Z^*}-W\frac{F^*Az^*(x)\beta}{(1+\beta Z^*)^2}+w(x)\left(\lambda+\frac{F^*A}{1+\beta Z^*}\right)-\int_{x_1}^{x_2}k(x,y)w(y)\, dy.\label{eigva2}
\end{align}
We also note that the stationary biomass of prey can be found from
\begin{eqnarray}\label{ssid2}
\mu=KA\frac{Z^*}{1+\beta Z^*}.
\end{eqnarray}

Utilising the above identity  we rewrite the eigenvalue problem (\ref{eigva1})-(\ref{eigva2}) as follows:
\begin{align}
0&=G\lambda+W\frac{F^*}{Z^*}\mu\left(\frac{\mu\beta}{KA}-1\right), \label{eigva3} \\
0&=G\frac{\mu}{K}\frac{z^*(x)}{Z^*}-W\frac{\beta}{A}\frac{\mu^2}{K^2}\frac{F^*}{Z^*}
\frac{z^*(x)}{Z^*}+w(x)\left(\lambda+\frac{\mu}{K}\frac{F^*}{Z^*}\right)-\int_{x_1}^{x_2}k(x,y)w(y)\,dy.\label{eigva4}
\end{align}
Any $\lambda$ for which there is a non-trivial vector $(G,w)$, which together with $\lambda$ solves equations \eqref{eigva3}-\eqref{eigva4}  is an eigenvalue. In general it is very difficult (if not impossible) to completely characterize the spectrum, mainly because of the inhomogeneous integral equation \eqref{eigva4}.
However, to show instability we only need to prove the existence of a solution $\lambda>0$ with a corresponding non-trivial eigenvector $(G,w)$. In what follows we establish a condition which guarantees the existence of a positive eigenvalue. To this end we utilise some well-known results from the spectral theory of positive and compact operators. For basic definitions and results not given here we refer the reader to \cite{Sch_1974}.

\begin{theorem}
\label{Theorem}
If the following inequality holds:
\begin{eqnarray}\label{instabcond}
\frac{AK}{\beta}-\mu+\frac{\mu}{K}\frac{F^*}{Z^*}\le \inf_{x\in [x_1,x_2]}\left\{\int_{x_1}^{x_2}k(x,y)\,dy\right\},
\end{eqnarray}
then the stationary state $(z^*,F^*)$ is unstable.
\end{theorem}

\begin{proof}
Assuming that $G\ne 0$ from equation \eqref{eigva3} we obtain:
\begin{eqnarray*}
\lambda=\frac{W}{G}\frac{F^*}{Z^*}\mu\left(1-\frac{\mu\beta}{KA}\right).
\end{eqnarray*}
Note that $1>\frac{\mu\beta}{KA}$ holds. With this, from equation \eqref{eigva4} we have:
\begin{eqnarray}\label{eigva5}
w(x)=\frac{\displaystyle\int_{x_1}^{x_2}k(x,y)w(y)\,dy+\frac{\mu}{K}\frac{z^*(x)}{Z^*}\left(\frac{\beta}{A}\frac{\mu}{K}\frac{F^*}{Z^*} W-G\right)}{\frac{W}{G}\frac{F^*}{Z^*}\mu\left(1-\frac{\mu\beta}{KA}\right)+\frac{\mu}{K}\frac{F^*}{Z^*}}, \quad x\in [x_1,x_2].
\end{eqnarray}
Hence for $G,W>0$ we introduce a parametrised family of bounded linear operators:
\begin{eqnarray}\label{opdef}
\mathcal{O}_{(G,W)} w=\frac{\displaystyle\int_{x_1}^{x_2}k(\cdot,y)w(y)\,d y+\frac{\mu}{K}\frac{z^*(\cdot)}{Z^*}\left(\frac{\beta}{A}\frac{\mu}{K}\frac{F^*}{Z^*} W-G\right)}{\frac{W}{G}\frac{F^*}{Z^*}\mu\left(1-\frac{\mu\beta}{KA}\right)+\frac{\mu}{K}\frac{F^*}{Z^*}},
\end{eqnarray}
with domain D$\left(\mathcal{O}_{(G,W)}\right)=L^1(x_1,x_2)$. We show that for some $G,W>0$ this (bounded) linear operator has eigenvalue $1$, namely its spectral radius, with a corresponding strictly positive eigenvector $w$.
To this end first we notice that in the region of the parameter plane determined by $W\ge G\frac{A}{\beta}\frac{K}{\mu}\frac{Z^*}{F^*}$, the operators $\mathcal{O}_{(G,W)}$ are positive and irreducible, since $k$ is strictly positive (see e.g. \cite{Sch_1974}). Moreover, since we assumed that the kernel $k$ is bounded above, i.e. $k(x,y)<\overline{k}$ for all $x,y\in [x_1,x_2]$, it is shown that the operators $\mathcal{O}_{(G,W)}$ are compact. We also note that along the half-line $W=G\frac{A}{\beta}\frac{K}{\mu}\frac{Z^*}{F^*},\, G>0$ the spectral radius $r\left(\mathcal{O}_{(G,W)}\right)$ is constant.
Next we are going to establish that at any point of the half-line $W=G\frac{A}{\beta}\frac{K}{\mu}\frac{Z^*}{F^*},\, G>0$ the spectral radius of the operator $\mathcal{O}_{(G,W)}$ is greater or equal than one. To this end recall (see e.g. \cite{Kato_1995}) that for any bounded linear operator $\mathcal{L}$ we have the following characterisation of the spectral radius:
\begin{eqnarray*}
r(\mathcal{L})=\lim_{n\to\infty}||\mathcal{L}^n||^{\frac{1}{n}}=\inf_{n\ge 1}||\mathcal{L}^n||^{\frac{1}{n}}.
\end{eqnarray*}
We have:
\begin{align}
\left|\left|\mathcal{O}_{\left(G,G\frac{A}{\beta}\frac{K}{\mu}\frac{Z^*}{F^*}\right)}\right|\right| & =\sup_{||w||_1=1}
\int_{x_1}^{x_2}\int_{x_1}^{x_2}\frac{k(x,y)}{\frac{AK}{\beta}-\mu+\frac{\mu}{K}\frac{F^*}{Z^*}}w(y)\,dy\,d x \nonumber \\
& \ge \frac{1}{x_2-x_1}\frac{1}{\frac{AK}{\beta}-\mu+\frac{\mu}{K}\frac{F^*}{Z^*}}\int_{x_1}^{x_2}\int_{x_1}^{x_2}k(x,y)\,dy\,dx \nonumber \\
& \ge \frac{1}{x_2-x_1}\frac{1}{\frac{AK}{\beta}-\mu+\frac{\mu}{K}\frac{F^*}{Z^*}}\int_{x_1}^{x_2} \inf_{x\in[x_1,x_2]}\left\{\int_{x_1}^{x_2}k(x,y)\,dy\right\}\,d x \nonumber \\
& \ge 1. \nonumber
\end{align}
Similarly we obtain:
\begin{align}
\left|\left|\mathcal{O}^2_{\left(G,G\frac{A}{\beta}\frac{K}{\mu}\frac{Z^*}{F^*}\right)}\right|\right| & =\sup_{||w||_1=1}
\int_{x_1}^{x_2}\int_{x_1}^{x_2}\frac{k(z,x)}{\frac{AK}{\beta}-\mu+\frac{\mu}{K}\frac{F^*}{Z^*}} \int_{x_1}^{x_2}\frac{k(x,y)}{\frac{AK}{\beta}-\mu+\frac{\mu}{K}\frac{F^*}{Z^*}}w(y)\,d y\,d x\,d z \nonumber \\
& \ge \frac{1}{x_2-x_1}\frac{1}{\left(\frac{AK}{\beta}-\mu+\frac{\mu}{K}\frac{F^*}{Z^*}\right)^2}\int_{x_1}^{x_2}\int_{x_1}^{x_2}k(z,x)\int_{x_1}^{x_2}k(x,y)\,dy\,dx \,d z \nonumber \\
& \ge \frac{1}{x_2-x_1}\frac{1}{\frac{AK}{\beta}-\mu+\frac{\mu}{K}\frac{F^*}{Z^*}}\int_{x_1}^{x_2} \int_{x_1}^{x_2}k(z,x)\,d x\,d z \nonumber \\
& \ge 1. \nonumber
\end{align}
In particular, for any $n\in\mathbb{N}$ as above we obtain $$\left|\left|\mathcal{O}^n_{\left(G,G\frac{A}{\beta}\frac{K}{\mu}\frac{Z^*}{F^*}\right)}\right|\right|\ge 1,$$  hence 
$$r\left(\mathcal{O}_{\left(G,G\frac{A}{\beta}\frac{K}{\mu}\frac{Z^*}{F^*}\right)}\right)\ge 1.$$
To establish the existence of a point in the region $W\ge G\frac{A}{\beta}\frac{K}{\mu}\frac{Z^*}{F^*}$ (in the positive quadrant of the parameter plane $(G,W)$), where the spectral radius equals $1$ it is convenient to re-parametrise the family of operators using polar coordinates. We write:
\begin{align*}
W=r\sin(\phi),\,G=r\cos(\phi),\quad \text{for}\quad r>0,\,\phi\in \left(0,\frac{\pi}{2}\right).
\end{align*}
With this new parametrisation we have:
\begin{equation}
\mathcal{O}_{(r,\phi)} w=\frac{\displaystyle\int_{x_1}^{x_2}k(\cdot,y)w(y)\,dy+\frac{\mu}{K}\frac{z^*(\cdot)}{Z^*}\left(\frac{\beta}{A}\frac{\mu}{K}\frac{F^*}{Z^*} r\sin(\phi)-r\cos(\phi)\right)}{\tan(\phi)\frac{F^*}{Z^*}\mu\left(1-\frac{\mu\beta}{KA}\right)+\frac{\mu}{K}\frac{F^*}{Z^*}}.
\end{equation}
Note that the condition $W\ge G\frac{A}{\beta}\frac{K}{\mu}\frac{Z^*}{F^*}$ corresponds to the condition $\tan(\phi)>\frac{A}{\beta}\frac{K}{\mu}\frac{Z^*}{F^*}$.
Also note that for every $r>0$ we have :
\begin{equation}
\lim_{\phi\to\frac{\pi}{2}}\left(\frac{\frac{\beta}{A}\frac{\mu}{K}\frac{F^*}{Z^*} r\sin(\phi)-r\cos(\phi)}{\tan(\phi)\frac{F^*}{Z^*}\mu\left(1-\frac{\mu\beta}{KA}\right)+\frac{\mu}{K}\frac{F^*}{Z^*}}\right)=0.
\end{equation}
Since we assumed that $k$ is bounded above by some $\overline{k}$ it follows that the spectral radius of the integral operator, which maps $w$ to $\int_{x_1}^{x_2}k(\cdot,y)w(y)\,d y$ is bounded above by $\overline{k}$, hence for every $r>0$
\begin{equation}
\lim_{\phi\to\frac{\pi}{2}}r\left(\mathcal{O}_{(r,\phi)}\right)=0.
\end{equation}
We note that we have shown above that condition \eqref{instabcond} implies that
\begin{equation*}
r\left(\mathcal{O}_{\left(r,\tan^{-1}\left(\frac{A}{\beta}\frac{K}{\mu}\frac{Z^*}{F^*}\right)\right)}\right)
=r\left(\mathcal{O}_{\left(G,G\frac{A}{\beta}\frac{K}{\mu}\frac{Z^*}{F^*}\right)}\right)\ge 1,
\end{equation*}
for any $r>0$. It follows from the Intermediate Value Theorem and the continuity of the spectral radius function (see e.g. \cite[Ch.IV.3.5]{Kato_1995}) that there is a point $(G_0,W_0)$ in the region $G,W>0,\, W\ge G\frac{A}{\beta}\frac{K}{\mu}\frac{Z^*}{F^*}$, where the spectral radius of $\mathcal{O}_{(G_0,W_0)}$ equals one. Note that the spectral radius itself is an eigenvalue. It follows from \cite[Theorem 6.6]{Sch_1974} that there is a strictly positive unique normalised eigenvector $v$  corresponding to the spectral radius $r\left(\mathcal{O}_{(G_0,W_0)}\right)$:
\begin{equation}
\mathcal{O}_{(G_0,W_0)}\, v=r\left(\mathcal{O}_{(G_0,W_0)}\right)\, v=1\cdot v.
\end{equation}
Hence we proved the existence of a solution:
\begin{equation*}
w(x)=W_0\,v(x),\,G_0,\, \lambda=\frac{W_0}{G_0}\frac{F^*}{Z^*}\mu\left(1-\frac{\mu\beta}{KA}\right)>0
\end{equation*}
of the eigenvalue problem \eqref{eigva1}-\eqref{eigva2}, which implies that the stationary state $(z^*,F^*)$ is unstable.
\end{proof}

Thus, the above theorem provides us with the sufficient condition of instability in the case of $a(x)\equiv A$, in general.

%%%%%%%%%%%%%%%%%%%%%%%%%%%%%%%%%%%%%%%%%%%%%%%%%%%%%%%%%%%%%%%%%%%%%%%%%%%%%%%%%%%%%%%%%%%%%%%%%%%%%%%%%%%%%%%%%%%%%%%%%%%%%%%%%%%%%%%%%%%%%%%%%%%%%%%%%%%%%%%%%%%%%%%%%%%%%%%%%%%%%%%%%%%%%%%%%%%%%%%%%%%%%%%

%%%%%%%%%%%%%%%%%%%%%%%%%%%%%%%%%%%%%%%%%%%%%%%%%%%%%%%%%%%%%%%%%%%%%%%%%%%%%%%%%%%%%%%%%%%%%%%%%%%%%%%

\section{Discussion and summary of the results}

Food-web models including rapid evolution are now gaining more and more popularity in the ecological literature in the recognition of the fact the rapid evolution can dra\-ma\-ti\-cally affect population dynamics \cite{Jones_2007, Jone_etal2009, Ellner_2011, Wolf_2012, Fussmann_Gonzalez2013, Tyutyunov2013}. In this paper, we have analytically investigated the recent eco-evolutionary predator-prey model proposed by Morozov \textit{et al.} in \cite{Morozov2013}, which is  given by equations (\ref{eq:Zoo})-(\ref{eq:Fish}). The previous numerical investigation shows a counter-intuitive result that rapid evolution of an organism's life trait combined with predation selectivity could stabilize this otherwise globally unstable system (when the carrying capacity of prey is considered to be infinitely large and the functional response of the predator is assumed to be a destabilizing Holling type II response). Since these important results  were obtained only using numerical simulation and only for some specific parametrization of the inheritance kernel $k(x,y)$ as well as the selectivity function $a(x)$, the natural question was about the generality of conclusions made in \cite{Morozov2013}. Here we  address the above question by investigating the model analytically and for generic model functions.

The main mathematical outcomes of our study are the following:

(i) We found conditions for the existence of the positive stationary states of model (\ref{eq:Zoo})-(\ref{eq:Fish}) for a class of kernels $k(x,y)$  given by (\ref{finite_rank}) and the arbitrary parametrization of the vulnerability to predation $a(x)$ .

(ii) We analytically derived stability conditions for the non-trivial stationary state in the case when the kernel is separable. In particular, when the kernel is constant and there is a perfect genetic mixing in the offspring in the prey population. The obtained characteristic equations allow us to construct the Hopf bifurcation curve without the need for a direct simulation of equations (\ref{eq:Zoo})-(\ref{eq:Fish}). Although, to construct the Hopf bifurcation curve one still need to solve transcendental equations (e.g. (\ref{Im_2})-(\ref{Re_2})), this task becomes substantially faster compared to direct simulations of the integro-differential equations.

(iii) For an arbitrary function $a(x)$, describing selectivity of predation, we analytically proved the stability of the coexistence stationary state in the case when the saturation of the predation $\beta$ as well as the length $h$ of the interval of variation of the life trait $[x_1, x_2]$ are small. Moreover, for simple kernels (e.g. $k=C, k=k(y)$ or $k=\bar{k} (y) k(x)$), by choosing  concrete parametrisations for the vulnerability to predation $a(x)$, we can easily explore the stability properties in the case when the interval $[x_1, x_2]$ is not small. In particular, our investigation shows that for $a(x)=a_0+a_1(x-x_1)$ and $a(x)=a_0+a_1(x-x_1)+a_1(x-x_1)^2$, the stationary state turns out to be always stable for a constant kernel $k$.

(iv) In contrast, for a constant $a(x)\equiv A$, we obtained that stabilisation is never possible in the case of $k=\bar{k} (y) k(x)$ or even in some more complicated cases of the separable kernel. For an arbitrary kernel $k=k(x,y)$ we derived sufficient conditions of instability with $a(x)\equiv A$ (see Theorem \ref{Theorem}).

(v) The framework of stability analysis and stationary states search suggested in this paper is rather generic and can be implemented to other eco-evolutionary models involving integro-differential equations.

From the biological perspective, the  analytical results obtained have clear interpretation. Stabilisation of the otherwise globally unstable predator-prey system is possible as a result of rapid evolution, and the necessary ingredients of stability are: (a) sufficiently large genetic variation within the prey population (for $\beta>0$ the length of the interval $[x_1, x_2]$ should be large enough to guarantee stability, see Figure 1); (b) there should be some selectivity of consumption of prey by its predator with respect to the life trait, i.e. $a=a(x)$, the sign of selectivity has no pronounced influence on stability (see Figure 1, as well as characteristic equation (\ref{linear_a2})); (c) the saturation in predation $\beta$ should not be too large. These requirements for stability are in a perfect agreement with the numerical results of \cite{Morozov2011, Morozov2013}.

We should mention here that our study of model (\ref{eq:Zoo})-(\ref{eq:Fish}) is by no means complete: several key issues are still need to be addressed. In particular, it would be important to prove (or disprove) the stability of the stationary state in the case $a(x)$ is an arbitrary positive function and when the saturation $\beta$ is small (see section 3.2 for details). Derivation of an exact condition for the type of Hopf bifurcation (i.e. subcritical or supercritical)  would be of a great practical use. Verifying the requirement that $a(x)$ should not be a constant to assure stability is an important challenge for future investigation. It would also be necessary to consider another class of demographic kernels $k(x,y)$, which are Gaussian-shaped and close to zero within large portion of the interval of $[x_1, x_2]$ except the area around their maximum. Biologically such kernels model the situation where the reproduction of the cohort with the life trait $x$ would be centred around $x$ with only small probability for offspring to deviate from this value. This represents a biologically relevant case of asexual reproduction. We are planning to address the above issues in future work.

Finally, it would be rather natural to incorporate diffusion into model \eqref{eq:Zoo}-\eqref{eq:Fish}. This would allow modelling of small random variations in the genetic trait in a deterministic fashion, for instance, this can account for random variation in the trait within the life time of an organism \cite{Morozov2013}. Structured population models with diffusion in the state space are attracting interest, see e.g. the recent papers \cite{BP_2007, Bouin_2012, CF_2012, Diekmann_2005,  Hadeler_2010, Lorz_2011, Michel_2013}. Since diffusion has a smoothing effect, in general, it would
be an interesting question to investigate whether the Hopf bifurcation shown here can be sustained in the analogue model with diffusion.
 
\vspace{1cm}
\section{Appendix}

Here we prove Theorem \ref{ssgeneral} stated in Section 2.3. In the special case of $\beta=0$ and $a\equiv A$ the steady state problem reads
\begin{align}
z^*(x) &=\frac{1}{F^*A}\int_{x_1}^{x_2}k(x,y)z^*(y)\, dy, \label{ssgeneq1} \\
\int_{x_1}^{x_2}z^*(x)\, dx=Z^* &=\frac{\mu}{KA}. \label{ssgeneq2}
\end{align}
Hence we define a parametrised family of bounded linear integral operators as follows
\begin{equation}
\mathcal{O}_{F}\,z=\frac{1}{FA}\int_{x_1}^{x_2}k(\cdot,y)z(y)\, dy,\quad F>0,
\end{equation}
with domain
\begin{align*}
D(\mathcal{O}_{F})=L^1(x_1,x_2),\quad F>0.
\end{align*}
Since the kernel $k$ is strictly positive and bounded, the integral operator $\mathcal{O}_{F}$ is irreducible and compact for every $F>0$, see e.g. \cite[Chapter V]{Sch_1974}, and the Fr\'{e}chet-Kolmogorov Theorem e.g. in \cite[Chapter X]{Yosida_1995}. Theorem 6.6 in \cite[Chapter V]{Sch_1974} implies that the spectral radius $r(\mathcal{O}_{F})$ of $\mathcal{O}_{F}$ is an isolated and dominant simple eigenvalue with a corresponding strictly positive eigenfunction. It also follows from the proof of Theorem 2.2 in \cite{Marek_1970} that the spectral radius is the only eigenvalue with a positive eigenvector. Also note that the spectral radius function $r\,: F\to r(\mathcal{O}_{F})$ is strictly monotone decreasing, and it is continuous. Continuity of the spectral radius function is a consequence of Theorem 3.16 in \cite[Chapter IV]{Kato_1995}. It follows that there exists a unique value $F^*$ for which the spectral radius $r(\mathcal{O}_{F^*})=1$ with a corresponding strictly positive eigenvector $\hat{z^*}$. That is we have
\begin{align*}
\mathcal{O}_{F^*}\, \hat{z^*}=1\cdot\hat{z^*}.
\end{align*}
We normalise this eigenvector $\hat{z^*}$, so that the stationary prey population density, which also satisfies equation \eqref{ssgeneq2}, is given by
\begin{align*}
z^*(x)=\frac{\mu}{KA}\frac{\hat{z^*}(x)}{\int_{x_1}^{x_2}\hat{z^*}(x)\, dx},\quad x\in [x_1,x_2].
\end{align*}

%%%%%%%%%%%%
%%%%%%%%%%%%%%%%%

%*******************************************************************
%Acknowledgements
%*******************************************************************

\vspace{10mm}
\noindent{\bf Acknowledgements}
 J. Z. Farkas was partly funded by a University of Stirling research and enterprise support grant. We thank Professors A. Gorban and S. Petrovskii (University of Leicester) for helpful discussions and comments. We also thank the anonymous referee for helpful suggestions and comments.

%\begin{acknowledgement}
% This work was funded by a grant from ...
%\end{acknowledgement}

\end{document}